\documentclass[a4paper,11pt]{article}

\usepackage[labelfont=bf,labelsep=period]{caption}
\usepackage[nodayofweek]{datetime}
\usepackage{enumitem}
\usepackage{float}
\usepackage[margin=2.5cm]{geometry}
\usepackage{graphicx}
\usepackage{mathtools}
\usepackage{numbertabbing}
\usepackage{times}
\usepackage{url}
\usepackage{xcolor}
\usepackage{xspace}

\setlist[]{topsep=2pt,partopsep=2pt,parsep=2pt,itemsep=2pt}

\floatstyle{ruled}
\newfloat{algo}{tbhp}{algo}
\floatname{algo}{Algorithm}

\newdateformat{simpledate}{\THEDAY\ \monthname[\THEMONTH]\ \THEYEAR}
\simpledate

\DeclarePairedDelimiter\ceil{\lceil}{\rceil}

\usepackage{amsmath} 
\allowdisplaybreaks[2]          

\usepackage{amssymb} 

\usepackage{amsthm} 
\newtheoremstyle{plain-boldhead}
  {\topsep}
  {\topsep}
  {\itshape}
  {}
  {\bfseries}
  {.}
  { }
  {\thmname{#1}\thmnumber{ #2}\thmnote{ (\bfseries #3)}}
\newtheoremstyle{definition-boldhead}
  {\topsep}
  {\topsep}
  {\normalfont}
  {}
  {\bfseries}
  {.}
  { }
  {\thmname{#1}\thmnumber{ #2}\thmnote{ (\bfseries #3)}}
\theoremstyle{plain-boldhead}
\newtheorem{theorem}{Theorem}

\newtheorem{corollary}[theorem]{Corollary}

\theoremstyle{definition-boldhead}

\newcommand{\str}[1]{\textsc{#1}}
\newcommand{\var}[1]{\textit{#1}}
\newcommand{\op}[1]{\textsl{#1}}
\newcommand{\msg}[2]{\ensuremath{\ifempty{#2} [\str{#1}] \else [\str{#1}, {#2}] \fi}}

\newcommand{\false}{\textsc{false}\xspace}
\newcommand{\true}{\textsc{true}\xspace}
\newcommand{\comment}[1]{}
\newcommand{\etal}{\emph{et al.}\xspace}

\def \ifempty#1{\def\temp{#1} \ifx\temp\empty }

\DeclareMathOperator*{\argmax}{arg\,max}
\newcommand\phase{\var{phase}\xspace}

\newcommand\ledger{\var{L}\xspace}
\newcommand\prevledger{\var{prevLedger}\xspace}
\newcommand\prevledgerID{\var{prevLedgerID}\xspace}
\newcommand\genesisledger{\var{genesisLedger}\xspace}
\newcommand\fullyledger{\var{validLedger}\xspace}
\newcommand\prevRoundTime{\var{prevRoundTime}\xspace}
\newcommand\roundTime{\var{roundTime}\xspace}
\newcommand\converge{\var{converge}\xspace}
\newcommand\open{\var{open}\xspace}
\newcommand\establish{\var{establish}\xspace}
\newcommand\accepted{\var{accepted}\xspace}
\newcommand\heartbeat{\var{heartbeat}\xspace}
\newcommand\agree{\var{agree}\xspace}
\newcommand\disagree{\var{disagree}\xspace}
\newcommand\iter{\var{i}\xspace}

\newcommand\threshold{\var{threshold}\xspace}
\newcommand\val{\var{$\sigma$}\xspace}
\newcommand\valCount{\var{valCount}\xspace}
\newcommand\disputes{\var{disputes}\xspace}
\newcommand\newvote{\var{newVote}\xspace}
\newcommand\pos{\var{prop}\xspace}
\newcommand\currPeerPositions{\var{currPeerProposals}\xspace}
\newcommand\validations{\var{validations}\xspace}
\newcommand\dt{\var{dt}\xspace}
\newcommand\result{\var{result}\xspace}
\newcommand\txns{\var{txns}\xspace}
\newcommand\position{\var{txns}\xspace}
\newcommand\proposal{\var{proposal}\xspace}
\newcommand\seq{\var{seq}\xspace}
\newcommand\nodeID{\var{node}\xspace}
\newcommand\tx{\var{tx}\xspace}
\newcommand\ourVote{\var{ourVote}\xspace}
\newcommand\node{\var{node}\xspace}

\newcommand\votes{\var{votes}\xspace}

\newcommand\set{\var{set}\xspace}

\newcommand\yays{\var{yays}\xspace}
\newcommand\nays{\var{nays}\xspace}
\newcommand\unl{\var{UNL}\xspace}
\newcommand\openTime{\var{openTime}\xspace}
\newcommand\id{\var{ID}\xspace}
\newcommand\parentid{\var{parentID}\xspace}

\newcommand\iterj{\var{j}\xspace}
\newcommand\s{\var{S}\xspace}
\newcommand\ol{\var{$\omega$}\xspace}
\newcommand\emptyhashmap{[\,]}
\newcommand\proposaltime{\var{time}\xspace}
\newcommand\support{\var{support}\xspace}

\newcommand\uncommitted{\var{uncommitted}\xspace}

\newcommand\tree{\var{tree}\xspace}

\begin{document}

\title{\bf Security Analysis of Ripple Consensus}

\author{Ignacio Amores-Sesar$^1$\\
  University of Bern\\
  \url{ignacio.amores@inf.unibe.ch}
  \and Christian Cachin$^1$\\
University of Bern\\
  \url{cachin@inf.unibe.ch}
  \and Jovana Mićić$^1$\\
University of Bern\\
  \url{jovana.micic@inf.unibe.ch}
}

\date{}

\footnotetext[1]{Institute of Computer Science, University of Bern,
  Neubr\"{u}ckstrasse 10, 3012 CH-Bern, Switzerland.}

\maketitle

\begin{abstract}\noindent
  The Ripple network is one of the most prominent blockchain platforms and
  its native XRP token currently has one of the highest cryptocurrency
  market capitalizations.  The \emph{Ripple consensus protocol} powers this
  network and is generally considered to a Byzantine fault-tolerant
  agreement protocol, which can reach consensus in the presence of faulty
  or malicious nodes.  In contrast to traditional Byzantine agreement
  protocols, there is no global knowledge of all participating nodes in
  Ripple consensus; instead, each node declares a list of other nodes that
  it trusts and from which it considers votes.

  Previous work has brought up concerns about the liveness and safety of
  the consensus protocol under the general assumptions stated initially by
  Ripple, and there is currently no appropriate understanding of its
  workings and its properties in the literature.  This paper closes this
  gap and makes two contributions.  It first provides a detailed, abstract
  description of the protocol, which has been derived from the source code.
  Second, the paper points out that the abstract protocol may violate
  safety and liveness in several simple executions under relatively benign
  network assumptions.
\end{abstract}

\section{Introduction}
\label{sec:intro}

Ripple is one of the oldest and most established blockchain networks; its
XRP token is ranked fourth in market capitalization in October 2020.  The Ripple
network is primarily aimed at fast global payments, asset exchange, and
settlement.  Its distributed consensus protocol is implemented by a
peer-to-peer network of validator nodes that maintain a history of all
transactions on the network~\cite{rippleledgeroverview20}.  Unlike
Nakamoto's consensus protocol~\cite{bitcoin} in Bitcoin or Ethereum, the
Ripple consensus protocol does not rely on ``mining,'' but uses a voting
process based on the identities of its validator nodes to reach
consensus.  This makes Ripple much more efficient than Bitcoin for
processing transactions (up to 1500 transactions per second) and lets it
achieve very low transaction settlement times (4--5 seconds).

\begin{figure}
  \centering \includegraphics[width=0.6\textwidth]{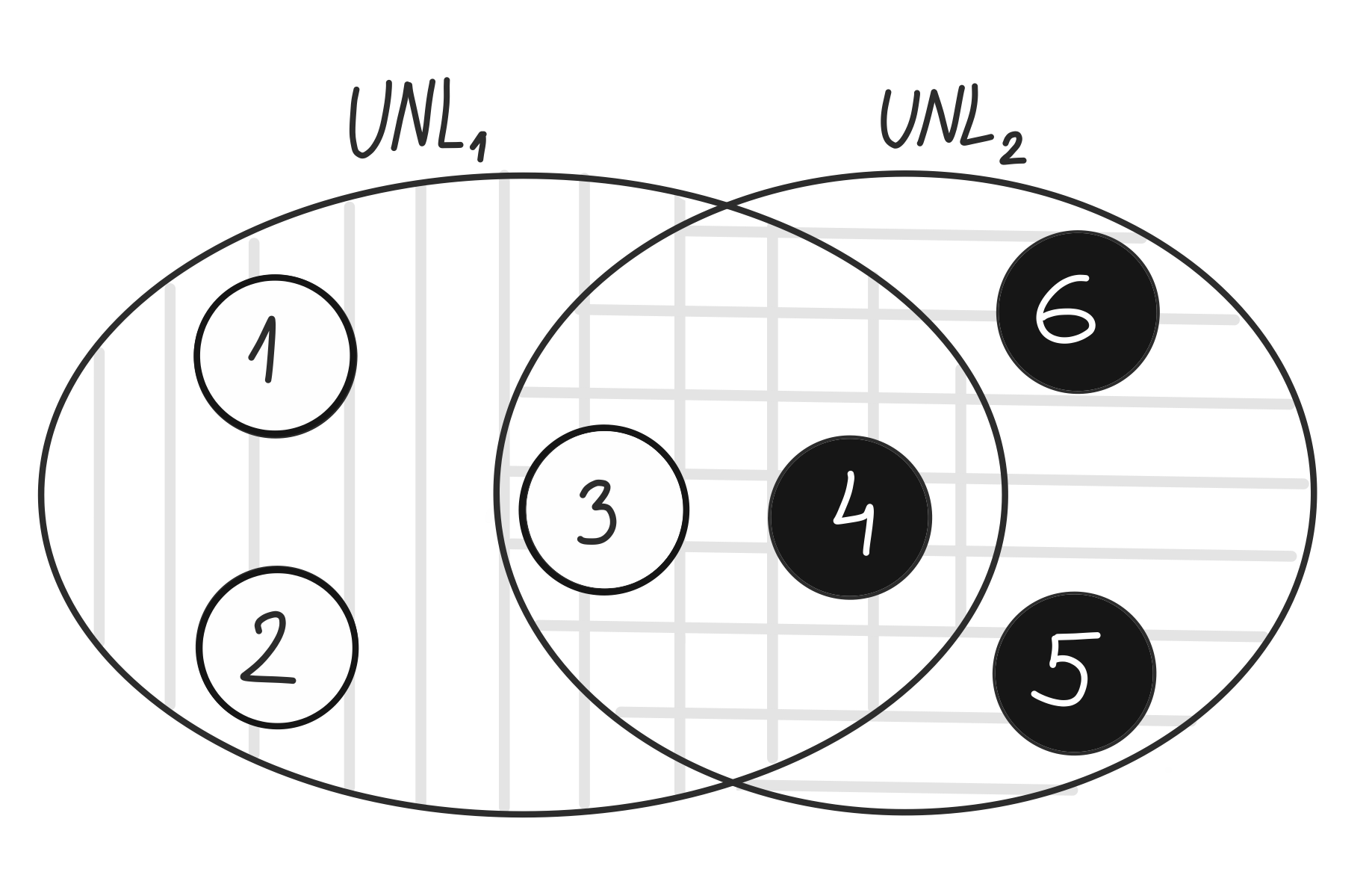}
  \caption{Example of a Ripple network configuration with six nodes and two
    UNLs, $\unl_1 = \{1, 2, 3, 4\}$  and
    $\unl_2 = \{3, 4, 5, 6\}$.
    Nodes 1, 2, and 3 (white) trust $\unl_1$, and nodes 4, 5, and 6
    (black) trust $\unl_2$.
    Notice that nodes 3 and 4 have more influence than the
    rest of nodes since they are in the intersection of both UNLs.
  }
  \label{fig:unls}
\end{figure}

However, Ripple's consensus protocol does not follow the established models
and algorithms for \emph{Byzantine agreement}~\cite{peshla80,lashpe82} or
\emph{Byzantine fault-tolerant (BFT) consensus}~\cite{caslis02}.  Those
systems start from a common set of nodes that are communicating with each
other to reach consensus and the corresponding protocols have been
investigated for decades.  Instead, the Ripple consensus protocol
introduces the idea of subjective validators, such that every node declares
some \emph{trusted validators} and effectively communicates only with those
nodes for reaching agreement on transactions.  With this mechanism, the
designers of Ripple aimed at opening up membership in the set of validator
nodes compared to BFT consensus.  The trusted validators of a node are
defined by a \emph{Unique Node List (UNL)}, which plays an important role
in the formalization of the protocol.  Every node maintains a static UNL in
its configuration file and considers only the opinions of nodes in its UNL
during consensus.  Figure~\ref{fig:unls} shows an example
network, where two UNLs are defined: $\unl_1 = \{1, 2, 3, 4\}$ and
$\unl_2 = \{3,4,5,6\}$; for instance, nodes 1, 2 and 3 may trust $\unl_1$,
and nodes 4, 5 and 6 may trust~$\unl_2$.

Consensus in Ripple aims at delivering the transactions
submitted by clients to all participating nodes in a common global order,
despite faulty or malicious (Byzantine) nodes~\cite{schwartz2014ripple}.
This ensures that the sequence of transactions, which are grouped into
so-called \emph{ledgers} and then processed by each node, is the same for
all nodes.  Hence, the states of all correct nodes remain synchronized,
according to the blueprint of state-machine replication~\cite{schnei90}.

Cachin and Vukoli\'{c}~\cite{cacvuk17} have earlier pointed out
that it is important to formally assess the properties of blockchain consensus
protocols.  Unfortunately, many systems have been designed and were
deployed without following the agreed-on principles on protocol analysis from
the literature.  Ripple is no exception to this, as we
show in this work.

Specifically, we focus on two properties that every sound protocol must
satisfy \cite{alpern1985defining}: \emph{safety} and \emph{liveness}.
Safety means that nothing ``bad'' will ever happen, and liveness means that
something ``good'' eventually happens.  Safety ensures that the network
does not fork or double-spend a token, for instance.  A violation of
liveness would mean that the network stops making progress and halts
processing transactions, which creates as much harm as forking.

This work first presents a complete, abstract description of the Ripple
consensus protocol (Section~\ref{sec:analysis}).  The model has been
obtained directly from the source code.  It is formulated in the language
spoken by designers of consensus protocols, in order to facilitate a better
understanding of the properties of Ripple consensus.  No formal description
of Ripple consensus with comparable technical depth has been available so
far (apart from the source itself).

Second, we exhibit examples of how safety and liveness may be violated in
executions of the Ripple consensus protocol (Sections~\ref{sec:safety} and
\ref{sec:liveness}).  In particular, the \emph{network may fork} under the
standard condition on UNL overlap stated by Ripple and in the presence of a
constant \emph{fraction} of Byzantine nodes.  The malicious nodes may simply send conflicting messages
to correct nodes and delay the reception of other messages among correct
nodes.  Furthermore, the consensus protocol may \emph{lose liveness} even if
all nodes have the same UNL and there is only \emph{one} Byzantine node.
If this would occur, the system has to be restarted manually.

Given these findings, we conclude that the consensus protocol of the Ripple
network is brittle and does not ensure consensus in the usual sense.
It relies heavily on synchronized clocks, timely
message delivery, the presence of a fault-free network, and an a-priori
agreement on common trusted nodes.  The role of the UNLs, their overlap,
and the creation of global consensus from subjective trust choices
remain unclear.  If Ripple instead had adopted a standard BFT
consensus protocol~\cite{CachinGR11}, as done by Tendermint~\cite{bukwmi18}, versions of
Hyperledger Fabric~\cite{abbccd18}, Libra~\cite{libra20librabft} or Concord~\cite{sbft}, then
the Ripple network would resist a much wider range of corruptions, tolerate
temporary loss of connectivity, and continue operating despite loss of
synchronization.

\section{Related work}
\label{sec:related}

Despite Ripple's prominence and its relatively high age among blockchain
protocols~--- the system was first released in 2012~--- there are only few
research papers investigating the Ripple consensus protocol compared to the
large number of papers on Bitcoin.  The original Ripple white paper of
2014~\cite{schwartz2014ripple} describes the UNL model and illustrates some
ideas behind the protocol.  It claims that under the assumption of
requiring an 80\%-quorum for declaring consensus, the intersection between
the UNLs of any two nodes $u$ and $v$ should be larger than 20\% of the
size of the larger of their UNLs, i.e.,
\[
  |\unl_u \cap \unl_v| \ \geq \ \frac{1}{5}\max\{|\unl_u|,|\unl_v|\}.
\]

The only earlier protocol analysis in the scientific literature of which we
are aware was authored by Arm\-knecht \etal in
2015~\cite{armknecht2015ripple}.  This work analyzes the Ripple consensus
protocol and outlines the security and privacy of the network compared to
Bitcoin.  The authors prove that a 20\%-overlap, as claimed in the white
paper, cannot be sufficient for reaching consensus and they increase the
bound on the overlap to at least 40\%, i.e.,
\[
  |\unl_u \cap \unl_v| \ > \ \frac{2}{5}\max\{|\unl_u|,|\unl_v|\}
\]

In a preprint of 2018, Chase and MacBrough~\cite{chase2018analysis} further
strengthen the required UNL overlap.  They introduce a high-level model of
the consensus protocol and describe some of its properties, but many
details appear unclear or are left out.  This work concludes that the
overlap between UNLs should actually be larger than 90\%.  The paper also
gives an example with 102 nodes that shows how liveness can be violated,
even if the UNLs overlap almost completely (by 99\%) and there are no
faulty nodes.  The authors conclude that manual intervention would be
needed to resurrect the protocol after this.

An analysis whose goal is similar to that of our work has been conducted by Mauri \etal~\cite{mauri2020formal}. Based on the source code, they give a verbal description of the consensus protocol, but do not analyze dynamic protocol properties. Our analysis, in contrast, provides a detailed, formal description with pseudocode and achieves a much better understanding of how the ``preferred ledger'' is chosen. Moreover, our work shows possible violations of safety and liveness, whereas Mauri \etal address only on the safety of the consensus protocol through sufficient conditions.

Other academic work mostly addresses network structure, transaction graph,
and privacy aspects of payments on the Ripple
blockchain~\cite{lumest17,mmskf18}, which is orthogonal to our focus.

\section{A description of the Ripple consensus protocol}
\label{sec:analysis}

The main part of our analysis consists of a detailed presentation of the
Ripple consensus protocol in this section and formally in
Algorithms~\ref{alg:alg1}--\ref{alg:alg4}.  Before we describe this, we
define the task that the protocol intends to solve.

\subsection{Specification}
\label{sec:specification}

Informally, the goal of the Ripple consensus protocol is ``to ensure that
the same transactions are processed and validated ledgers are consistent
across the peer-to-peer XRP Ledger network''~\cite{rippledoc20}.  More
precisely, this protocol implements the task of synchronizing the nodes so
that they proceed through a common execution, by appending successive
\emph{ledgers} to an initially empty history and where each ledger consists
of a number of \emph{transactions}.  This is the problem of replicating a
service in a distributed system, which goes back to Lamport et al.'s
pioneering work on Byzantine agreement~\cite{peshla80,lashpe82}.  The
problem has a long history and a good summary can be found in the book ``30-year
perspective on replication''~\cite{CharronBostPS10}.

For replicating an abstract service among a set of nodes, the service is
formulated as a deterministic state machine that executes
\emph{transactions} submitted by clients or, for simplicity, by the nodes
themselves.  The \emph{consensus protocol} disseminates the transactions
among the nodes, such that each node locally executes the \emph{same
  sequence of transactions} on its copy of the state.  The task provided by
this protocol is also called \emph{atomic broadcast}, indicating that
the nodes actually disseminate the transactions.  When each node locally
executes the same sequence of transactions, as directed by the protocol,
and since each transaction is deterministic, all
nodes will maintain the same copy of the state~\cite{schnei90}.

More formally, \emph{atomic broadcast} is characterized by two events
dealing with transactions: \op{submission} and \op{execution}, which may
each occur multiple times.  Every node may submit a transaction $\tx$ by
invoking $\op{submit}(\tx)$ and atomic broadcast applies $\tx$ to the
application state on the node through $\op{execute}(\tx)$.  A protocol for
atomic broadcast then ensures these properties~\cite{hadtou93,CachinGR11}:

\begin{description}
\item[Validity:] If a correct node~$p$ submits a transaction~$\tx$, then $p$
  eventually executes~$\tx$.
\item[Agreement:] If a transaction~$\tx$ is executed by some correct node,
  then $\tx$ is eventually executed by every correct node.
\item[Integrity:] No correct node executes a transaction more than once;
  moreover, if a correct node executes a transaction~$\tx$ and the
  submitter~$p$ of $\tx$ is correct, then $\tx$ was previously submitted
  by~$p$.
\item[Total order:] For transactions $\tx$ and $\tx'$, suppose $p$ and $q$
  are two correct nodes that both execute $\tx$ and $\tx'$.  Then $p$
  executes $\tx$ before $\tx'$ if and only if $q$ executes $\tx$ before
  $\tx'$.
\end{description}

Our specification does not refer to the heterogeneous trust structure
defined by the UNLs and simply assumes all nodes should execute the same
transactions.  This corresponds to the implicit assumption in Ripple's code
and documentation.  We note that the question of establishing global
consistency in a distributed system with subjective trust structures is a
topic of current research, as addressed by asymmetric quorum
systems~\cite{cactac19} or in the context of Stellar's
protocol~\cite{logama19}, for example.

\subsection{Overview}
\label{sec:overview}

The following description was obtained directly from the source code.  Its
overall structure retains many elements and function names found in the
code, so that it may serve as a guide to the source for others and to
explain its working.  If the goal had been to compare Ripple consensus to
the existing literature on synchronous Byzantine agreement protocols, the
formalization would differ considerably.

The protocol is highly \emph{synchronous} and relies on a common notion of
time.  It is structured into successive \emph{rounds of consensus}, whereby
each round agrees on a \emph{ledger} (a set of transactions to execute).
Each round roughly takes a predefined amount of time and is driven by a
heartbeat timer, which triggers a state update once per second.  This
contrasts with the Byzantine consensus protocols with partial
synchrony~\cite{dwlyst88}, such as PBFT~\cite{caslis02}, which can tolerate
arbitrarily long periods of asynchrony and rely on clocks or timeouts only
for liveness.  The Ripple protocol aims to agree on a transaction set
within each synchronized round.  The round ends when all nodes collectively
declare to have reached consensus on a \emph{proposal} for the round.  The
protocol is then said to \emph{close} and later \emph{validate} a ledger
containing the agreed-on transaction set.  However, the transactions in the
ledger are executed only after another protocol step, once the ledger has
become \emph{fully validated}; this occurs in an asynchronous process in
the background.  Transaction execution is only logically synchronized with
the consensus round.

A \emph{ledger} consists of a batch of transactions that result from a
consensus round and contains a hash of the logically preceding ledger.
Ledgers are stored persistently and roughly play the role of blocks in
other blockchain protocols.  Each node locally maintains three different
ledgers: the \emph{current ledger}, which is in the process of building
during a consensus round, the \emph{previous ledger}, representing the most
recently closed ledger and the \emph{valid ledger}, which is the last fully
validated ledger in the network.

In more detail, a consensus round has three \emph{phases}: \open,
\establish, and \accepted.  According to the the state diagram shown in
Figure~\ref{fig:phases}, the usual phase transition goes from \open to
\establish to \accepted and then proceeds to the next consensus round,
which starts again from \open.  However, it is also possible that the phase
changes from \establish to \open, if a node detects that it has been forked
from the others to a wrong ledger and resumes processing after
switching to the ledger agreed by the network.

\begin{figure}
  \centering
  \includegraphics[width=0.7\textwidth]{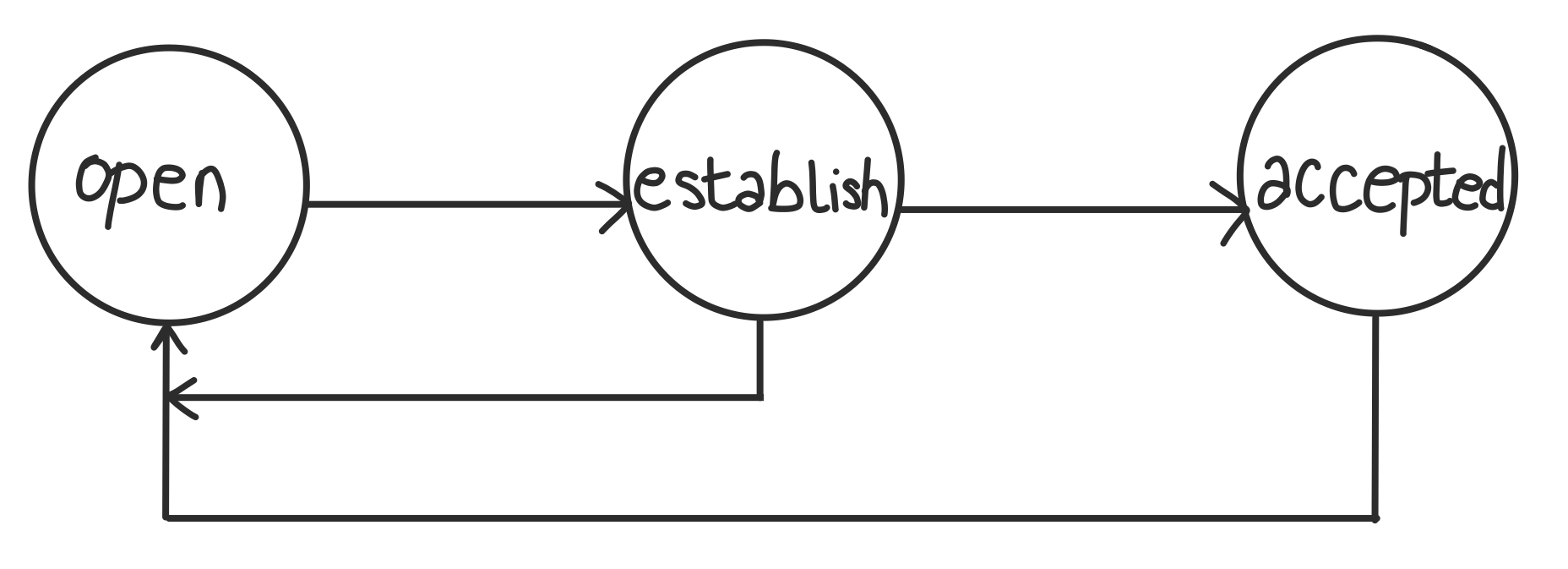}
  \caption{Phases and transitions between phases of a consensus round.}
  \label{fig:phases}
\end{figure}

Nodes may submit transactions at any time, concurrently to executing the
consensus rounds.  They are disseminated among the nodes through a
\emph{gossip layer} that ensures only weak consistency.  All transactions
that have been received from gossip are placed into a buffer.
Apparently, the original design assumed that the gossip layer ensures a
notion of consistency that prevents Byzantine nodes from equivocating, in
the sense of correct nodes never receive different messages from them.
This assumption has been dropped later~\cite{chase2018analysis}.

The protocol rounds and their phases are implemented by a state machine,
which is invoked every second, when the global \heartbeat timer ticks.
Messages from other nodes are received asynchronously in the background and
processed during the next timer interrupt.

The timeout handler (L\ref{heartbeat}) first checks if the local
\emph{previous ledger} is the same as the \emph{preferred ledger} of a
sufficient majority of the nodes in the network.  If not, the node has been
forked or lost synchronization with the rest of the network and must bring
itself back to the state agreed by the network.  In this case, it starts a
new consensus round from scratch.

When the node enters a new round of consensus, it sets the phase to \open,
resets round-specific data structures, and simply waits for the buffer to
fill up with submitted transactions.  Once the node has been in the \open
phase for more than half of the duration of the previous consensus round,
the node moves to the \establish phase
(L\ref{phaseestablish}--L\ref{closeledger}; function \op{closeLedger}).  It
locally closes the ledger, which means to initialize its proposal for the
consensus round and to send this to the other nodes in its UNL.

During the \establish phase, the nodes exchange their proposals for the
transactions to decide in this consensus round (using \str{proposal}
messages).  Obviously, these proposals may contain different transaction
sets.  All transactions on which the proposals from other nodes differ
become \emph{disputed}.  Every node keeps track of how many other nodes in
its UNL have proposed a disputed transaction and represents this
information as \emph{votes} by the other nodes.  The node may remove a
disputed transaction from its own proposal, or add one to its proposal,
based on the votes of the others and based on the time that has passed.
Specifically, the node increases the necessary threshold of votes for
changing its own vote on a disputed transaction depending on the duration
of the \establish phase with respect to the time taken by the previous
consensus round.

The node leaves the \establish phase when it has found that there is a
consensus on its proposal (L\ref{loopend}--L\ref{onaccept}; functions
\op{haveConsensus} and \op{onAccept}).  The node constructs the next ledger
(the ``last closed ledger'') by ``applying'' the decided transactions.
This ledger is signed and broadcast to the other nodes in a
\str{validation} message.

The node then moves to the \accepted phase and immediately initializes a
new consensus round.  Concurrently, the node receives \str{validation}
messages from the nodes in its UNL.  It verifies them and counts how many
other nodes in its UNL have issued the same validation.  When this number
reaches 80\% of the nodes in its UNL, the ledger becomes fully validated
and the node executes the transactions contained in it.

\subsection{Details}
\paragraph{Phase open.}
Function \op{beginConsensus} starts a consensus round for the next ledger
(L\ref{beginConsensusFunction}). Each ledger (L\ref{typeledger}) contains a
hash (\id) that serves as its identifier, a sequence number (\seq), a hash
of the parent ledger (\parentid), and a transaction set (\txns), denoting
the transactions applied by the ledger.

The node records the time when the \open phase started (\openTime,
L\ref{resetopentime}), so that it can later calculate how long the \open
phase has taken. This is important because the duration of the \open phase
determines when to close the ledger locally. If the time that has passed
since \openTime is longer or equal to half of the previous round time
(\prevRoundTime), consensus moves to phase \establish by calling the
function \op{closeLedger} (L\ref{closeledger}).  Meanwhile all nodes submit
transactions with the gossip layer (L\ref{submit}) and each node stores the
transaction received via gossip messages in its transaction set~\s
(L\ref{recievesubmit}).  We model transactions as bit strings.  In some
places, and as in the source code, we use a short, unique \emph{transaction
  identifier} (of type int) for each transaction $\tx \in \{0,1\}^*$,
computed by a function $\op{TxID}(\tx)$. A transaction set is a set of
binary strings here, but the source code maintains a transaction set using
a hash map, containing the transaction data indexed by their identifiers.


\paragraph{Phase establish.}
When the node moves from \open to \establish, it calls \op{closeLedger}
that creates an initial proposal (stored in \result.\proposal), containing
all transactions received from the gossip layer (L\ref{createproposal})
that have not been executed yet.  A proposal structure
(L\ref{typeproposal}) contains the hash of the previous ledger
(\prevledgerID), a sequence number (\seq), the actual set (\txns) of
proposed transactions (in the source code named \textit{position}), an
identifier of the node (\node) that created this proposal, and a timestamp
(\proposaltime) when this proposal is created (L\ref{createproposal}).

The node then broadcasts the new proposal as a \str{proposal} message
(L\ref{broadcastproposal1}) to all nodes in its UNL.  When they receive it,
they will store its contents in their \currPeerPositions collection of
proposals (L\ref{receiveprop1}), if the message originates from a node
their respective UNL.  The \op{closeLedger} function also sets
\result.\roundTime to the current time (L\ref{roundtimenow}).  This serves
to measure the duration of the \establish phase and will be used later to
determine how far the consensus process has converged.

Based on the proposals from other nodes, each node computes a set of
disputed transactions (L\ref{createDisputesFunction}).  A disputed
transaction (DisputedTx, L\ref{typedisputedtx}) contains the transaction
itself (\tx), a binary vote (\ourVote) by the node on whether this
transaction should be included in the ledger, the number of ``yes'' and
``no'' votes from other nodes on the transaction (\yays and \nays), taken
from their \str{proposal} messages, and the list of votes on this
transaction from the other nodes (\votes).

A transaction becomes disputed when it is proposed by the node itself and
some other node does not propose it, or vice versa.  The node determines
these by comparing its own transaction set with the transaction sets of all
other nodes~(L\ref{creatediff}).  Every disputed transaction is recorded
(as a DisputedTx structure) in the collection \result.\disputes
(L\ref{createdt}--L\ref{cdend}).

During the \establish phase, the node constantly updates its votes on all
disputed transactions (L\ref{callupdateourproposals};
L\ref{updateourproposals}; L\ref{updateVoteFunction}) for responding to
further \str{proposal} messages that have been received.  A vote may change
based on the number of nodes in favor of the transaction, the
\emph{convergence ratio} (\converge) and a \threshold.  Convergence
measures the expected progress in one single consensus round and is
computed from the duration of the \establish phase, the duration of the
previous round, and an assumed maximal consensus-round time
(L\ref{convergecalc}). The value for the threshold is predefined. The
further the consensus converges, the higher is the threshold that the
number of opposing votes needs to reach so that the node changes its own
vote~(L\ref{l:newvote}).  Whenever the node's proposal is updated, the node
broadcasts its new proposal to the other nodes
(L\ref{broadcastnewproposal}) and the disputed transactions are
recomputed~(L\ref{updateDisputes}).

Afterwards, the node checks if consensus on its proposed transaction set
\result.\txns is reached, by calling the function \op{haveConsensus}
(L\ref{loopend}). The node counts agreements (L\ref{l:agree}) and
disagreements (L\ref{l:disagree}) with \result.\txns.  If the fraction of
agreeing nodes is at least 80\% with respect to the UNL
(L\ref{checkhaveconsensus}), then consensus is reached.  The node proceeds
to the \accepted phase by calling the function
\op{onAccept}~(L\ref{onaccept}).


\paragraph{Phase accepted.}

The function \op{onAccept} (L\ref{onacceptfunction}) ``applies'' the
agreed-on transaction set and thereby creates the next ledger (called the
``last closed ledger'' in the source code; L\ref{l:buildlcl}). This ledger
is then signed (L\ref{l:sign}) and broadcast to the other nodes as a
\str{validation} message (L\ref{broadcastvalidation}). This marks the end
of the \accepted phase and a new consensus round is initiated by the node
(L\ref{l:startconsensusagain}).

Meanwhile, in the background, the node receives \str{validation} messages
from other nodes in its UNL and tries to verify them (L\ref{receiveval}).
This verification checks the signature and if the sequence number of the
received ledger is the same as the sequence number of the own ledger.  All
validations that satisfy both conditions and contain the node's own
agreed-on ledger are counted (L\ref{l:valcount}); this comparison uses the
cryptographic hash of the ledger structure in the source code.  Again, if
80\% of nodes have validated the same ledger and if the sequence number of
that ledger is larger than that of the last fully validated
ledger~(L\ref{l:checkquorum}), the ledger becomes fully validated
(L\ref{l:fullyledger}).  The node then executes the transactions in the
ledger~(L\ref{executetxns}).  In other words, the consensus decision has
become final.


\paragraph{Preferred ledger.}

A node participating in consensus regularly computes the \emph{preferred
  ledger}, which denotes the current ledger on which the network has
decided.  Due to possible faults and network delays, the node's \prevledger
may have diverged from the preferred ledger, which is determined by calling
the function \op{getPreferred(\var{validLedger})}
(L\ref{getPreferredFunction}).  Should the network have adopted a different
ledger than the \prevledger of the node, the node switches to this ledger
and restarts the consensus round with the new ledger.

Notice that the validated ledgers from all correct nodes form a tree,
rooted in the initial ledger (\genesisledger).  Each node stores all valid
ledgers that it receives in a tree-structured variable \tree.  Whenever the
node receives a \str{validation} message containing a ledger $\ledger'$, it
adds $\ledger'$ to \tree (L\ref{addtotree}).  In order to compute the
preferred ledger, we define the following functions, which are derived from
the ledgers in \tree and in the received \str{validation} messages:
\begin{itemize}
\item $\op{tip-support}(L)$ for a ledger $L$ is the number of validators in
  the \unl that have validated~$L$. In other words,
  \[
    \op{tip-support}(\ledger) = \big|\{\ j\in\unl\ |\
    \validations[j]=\ledger\}\big|.
  \]
\item $\op{support}(L)$ for a given ledger~\ledger is the sum of the tip
  support of $L$ and all its descendants in \tree, i.e.,
  \[
    \op{support}(L) = \op{tip-support}(L) +
      \sum_{L'\ \text{is a child of}\ L \ \text{in}\ \tree}\op{tip-support}(L^\prime).
  \]
\item $\uncommitted(L)$ for a ledger~\ledger denotes the number of
  validators whose last validated ledger has a sequence number that is
  strictly smaller than the sequence number of \ledger.  More formally,
  \[
    \uncommitted(L) =
    \big| \{j\in\unl\;|\; \validations[j].\seq < \ledger.\seq\} \big|.
  \]
\end{itemize}

With these definitions, we now explain how \op{getPreferred(Ledger
  \ledger)} proceeds (L\ref{getPreferredFunction}--L\ref{getPreferredEnd}).
If $L$ has no children in \tree, it returns $L$ itself.  Otherwise, the
function considers the child of $L$ that has the highest support among all
children ($M$).  If the support of $M$ is still smaller than the number of
validators that are yet uncommitted at this ledger-sequence number,
then $L$ is still the preferred ledger~(L\ref{preferredL1}).
Otherwise, if the support of $M$ is guaranteed to exceed the support of any
of its siblings~$N$, even when the uncommitted validators would also
support~$N$, then the function recursively calls $\op{getPreferred}$ on~$M$, which outputs the preferred
ledger for~$M$ and returns this as the preferred ledger for~\ledger. Otherwise, $L$
itself is returned as the preferred ledger. Observe that in the case when $M$ has no siblings conditions in L\ref{cond1} and L\ref{cond2} are equivalent. Then is enough to check if support of $M$ is greater than uncommitted od $M$.

\paragraph{Functions.} For the simplicity of pseudocode, there are some functions that are not fully explained. These functions are:
\begin{itemize}
  \item \op{startTimer}(\var{timer}, \var{duration}) starts \var{timer}, which expires after the time passed as \var{duration}.
  \item \op{clock.now()} returns the current time.
  \item \op{Hash} creates a unique identifier (often denoted \op{ID}) of a data structure by converting the data to a canonical representation and applying a cryptographic hash function to this.
  \item $\mathcal{A} \; \triangle\; \mathcal{B}$ denotes the symmetric set difference.
  \item $\op{boolToInt}(b)$ converts a logical value~$b$ to an integer
  and returns \op{$b$? 0:1}.
  \item \op{sign\textsubscript{\iter}(\ledger)} creates a cryptographic digital signature for ledger~\ledger by node~\iter.
  \item \op{verify\textsubscript{\iter}}\((\ledger,\val)\) checks if the digital signature on \ledger from node~\iter is valid.
  \item \op{siblings(M)} returns the set of nodes, different from $M$, that have the same parent as $M$. 
\end{itemize}

\begin{algo*}
  \vbox{
  \small
  \begin{numbertabbing}\reset
  xxxx\=xxxx\=xxxx\=xxxx\=xxxx\=xxxx\=xxxx\kill
  \textbf{Type} \label{} \\
  \> Enum $\text{Phase} = \{\open, \establish, \accepted\}$ \label{} \\
  \> $\text{Tx} = {\{0,1\}^*}$ \` // a transaction \label{} \\
  \> $\text{TxSet} = 2^{\text{Tx}}$ \label{} \\
  \> DisputedTx( \` // DisputedTx.h:50 \label{typedisputedtx} \\
  \>\> Tx \tx, \` // disputed transaction \label{} \\
  \>\> \text{bool} $\ourVote$, \` // binary vote on whether transaction should be included \label{} \\
  \>\> int \yays, \` // number of yes votes from others \label{} \\
  \>\> int \nays, \` // number of no votes from others \label{} \\
  \>\> HashMap[int $\to$ bool] \votes) \` // collection of votes indexed by node \label{} \\
  \> Ledger( \` // Ledger.h:77 \label{typeledger} \\
  \>\> Hash \id,  \` // identifier \label{} \\
  \>\> int \seq, \` // sequence number of this ledger  \label{} \\
  \>\> Hash \parentid, \` // identifier of ledger's parent  \label{} \\
  \>\> TxSet \txns) \` // set of transactions applied by ledger \label{} \\
  \> Proposal( \` // ConsensusProposal.h:52 \label{typeproposal} \\
  \>\> Hash \prevledgerID, \` // hash of the previous ledger, on which this proposal builds \label{} \\
  \>\> int \seq, \` // sequence number \label{} \\
  \>\> TxSet \position, \` // proposed transaction set, called \var{position} at ConsensusProposal.h:73 \label{} \\
  \>\> int \nodeID, \` // node that proposes this \label{} \\
  \>\> milliseconds \proposaltime) \` // time when proposal is created \label{} \\
  \> ConsensusResult( \` // ConsensusTypes.h:201 \label{} \\
  \>\> TxSet \txns, \` // set of transactions consensus agrees on \label{} \\
  \>\> Proposal \proposal, \` // proposal containing transaction set \label{} \\
  \>\> HashMap[int $\to$ DisputedTx] \disputes, \` // collection of disputed transactions \label{} \\
  \>\> milliseconds \roundTime) \` // duration of the establish phase \label{} \\
  \\  
  \textbf{State} \label{} \\
  \> Phase \(\phase\)\` // phase of the consensus round for agreeing on one ledger \label{} \\
  \> Tree \tree \` // tree representation of received valid ledgers\label{} \\
  \> Ledger \(\ledger\)\` // current working ledger \label{} \\
  \> Ledger \(\prevledger\)\` // last agreed-on (``closed'') ledger according to the network \label{} \\
  \> Ledger \(\fullyledger\)\` // ledger that was most recently fully validated by the node \label{} \\
  \> TxSet \(\s\)\` // transactions submitted by clients that have not yet been executed \label{} \\
  \> ConsensusResult \(\result\)\` // data relevant for the outcome of consensus on a single ledger \label{} \\
  \> HashMap[int $\to$ Proposal] \(\currPeerPositions\)\` // collection of proposals indexed by node \label{} \\
  \> HashMap[int $\to$ Ledger] \(\validations \)\` // collection of validations indexed by node \label{} \\
  \> milliseconds \(\prevRoundTime\)\` // time taken by the  previous consensus round, initialized to 15s \label{} \\ 
  \> float \(\converge\) $\in \left[0, 1 \right]$\` // ratio of round time to \prevRoundTime \label{} \\
  \> \(\unl \subseteq \{1, \dots, M\}\) \` // validator nodes trusted by node $i$, taken from the configuration file \label{} \\
  \> milliseconds \(\openTime\)\` // time when the last \open phase started \label{} \\

  \\
  \textbf{function} \op{initialization()} \label{} \\
  \> \(\prevledger \gets \genesisledger\) \` // \genesisledger is the first ledger in the history of the network \label{} \\
  \> $S \gets \{\}$ \label{} \\
  \> \op{beginConsensus()} \` // start the first round of consensus \label{} \\
  \> $\op{startTimer}(\heartbeat, 1 s)$ \` // NetworkOPs.cpp:673 \label{initheartbeat} \\
  \\
  \textbf{upon} submission of a transaction \tx \textbf{do} \label{submit} \\
  \> send message \(\msg{submit}{\tx}\) with the gossip layer  \label{} \\
  \\
  \textbf{upon} receiving a message \(\msg{submit}{\tx}\) from the gossip layer \textbf{do} \label{recievesubmit} \\
  \> \( \s \gets \s \cup \{\tx\} \) \label{} 
  \end{numbertabbing}
  }
  \caption{{\bf Ripple consensus protocol for node \iter (continues on next pages)}}
  \label{alg:alg1} 
\end{algo*}

\begin{algo*}
  \vbox{
  \small
  \begin{numbertabbing}
  xxxx\=xxxx\=xxxx\=xxxx\=xxxx\=xxxx\=xxxx\kill
  \textbf{function} \op{beginConsensus()} \` // start a new round of consensus, Consensus.h:663 \label{beginConsensusFunction} \\
  \> \(\phase \gets \open\) \` // Consensus.h:669\label{phaseopen} \\
  \> \(\result \gets (\{\}, \bot, \emptyhashmap, 0)\) \` // Consensus.h:674 \label{} \\
  \> \(\converge \gets 0\) \` // Consensus.h:675 \label{} \\
  \> \(\openTime \gets \op{clock.now()}\) \` // remember the time when this consensus round started \label{resetopentime} \\
  \> \(\currPeerPositions \gets [\,]\) \` // reset the proposals for this consensus round \label{emptylist} \\ 
  \\
  \textbf{upon} \(\op{timeout}(\heartbeat)\) \textbf{do} \` // Consensus.h:818\label{heartbeat} \\
  \> \(\ledger' \gets \op{getPreferred}(\var{validLedger}) \) \label{} \\
  \> \textbf{if} $\ledger' \neq \prevledger$ \textbf{then} \label{} \\
  \>\> $\prevledger \gets \ledger'$ \label{} \\
  \>\> \op{beginConsensus(\prevledger)} \label{} \\
  \> \textbf{if} \phase = \open \textbf{then}
  \` // wait until the closing ledger can be determined locally, Consensus.h:829 \label{} \\
  \>\> \textbf{if} $ (\op{clock.now()} - \openTime) \geq \frac{\prevRoundTime}{2}$ \textbf{then} \` // Consensus.cpp:75 \label{ifdelta} \\
  \>\>\> \(\phase \gets \establish\) \label{phaseestablish} \\
  \>\>\> \op{closeLedger()}\` // initialize consensus value in \result \label{closeledger} \\
  \> \textbf{else if} \phase = \establish \textbf{then} \` // agree on the contents of the ledger to close, Consensus.h:833 \label{loopstart} \\
  \>\> $\result.\roundTime \gets \op{clock.now()} - \result.\roundTime$ \label{} \\
  \>\> \(\converge \gets \frac{\op{\result.\roundTime}}{\max \{\prevRoundTime, 5 s\}}\) \label{convergecalc} \\
  \>\> \op{updateOurProposals()} \` // update consensus value in \result \label{callupdateourproposals} \\
  \>\> \textbf{if} \op{haveConsensus()} \textbf{then} \label{loopend} \\
  \>\>\> \(\phase \gets \accepted\) \label{l:phaseaccepted} \\
  \>\>\> \op{onAccept()}
  \` // note this immediately sets $\phase = \open$ inside \op{beginConsensus()} \label{onaccept} \\
  \> \textbf{else if} \phase = \accepted \textbf{then} \` // Consensus.h:821 \label{} \\
  \>\>  // do nothing \label{} \\
  \> $\op{startTimer}(\heartbeat, 1 s)$ \label{} \\
  \\
  // transition from \open to \establish phase \label{} \\
  \textbf{function} \op{closeLedger()}\` // Consensus.h:1309 \label{} \\
  \> \(\ledger \gets (\bot,\prevledger.seq + 1,\bot,\{\})\) \label{intiLedger} \\
  \> \(\result.\txns \gets \s\) \` // propose the current set of submitted transactions  \label{applyheldtxns} \\
  \> \(\result.\proposal \gets (\op{Hash(\prevledger)}, 0, \result.\txns, \iter, \op{clock.now())}\) \label{createproposal} \\
  \> $\result.\roundTime \gets \op{clock.now()}$ \label{roundtimenow} \\
  \> broadcast message \(\msg{proposal}{\result.\proposal}\) \label{broadcastproposal1} \\
  \> \(\result.\disputes \gets \emptyhashmap\) \` // a dispute exists for a transaction not proposed by all nodes in the UNL \label{} \\
  \> \textbf{for} \(j \in \unl \) \textbf{such that} $\currPeerPositions[j] \neq \bot$ \textbf{do} \label{createdisputes} \\
  \>\> $\op{createDisputes}(\currPeerPositions[j].\position)$ \` // compared to \result.\txns, Consensus.h:1334 \label{} \\
  \\
  \textbf{upon} receiving a message \(\msg{proposal}{\pos}\) \textbf{such that} $\pos = (nl, \cdot, \cdot, j, \cdot)$ \textbf{and} \label{receiveprop1} \\
  \>\> $j \in \unl$ \textbf{and} $nl = \op{Hash(\prevledger)}$ \textbf{do} \label{} \\
  \> \(\currPeerPositions[j] \gets \pos\) \` // Consensus.h:781 \label{recieveprop2} \\
  \\
  \textbf{function} \op{createDisputes(TxSet \set)} \` // Consensus.h:1623 \label{createDisputesFunction} \\
  \> \textbf{for} $\tx \in \result.\txns \; \triangle \; \set$ \textbf{do} \` // all transactions that differ between \result.\txns and \set \label{creatediff} \\
  \>\> \(\dt \gets \bigl(\tx, (\tx \in \result.\txns), 0, 0, \emptyhashmap\bigr)\) \` // $\dt$ is a disputed transaction \label{createdt} \\
  \>\> \textbf{for} $k \in \unl$ \textbf{such that} $\currPeerPositions[k] \neq \bot$ \textbf{do} \label{cdstart} \\
  \>\>\> \textbf{if} \(\tx \in \currPeerPositions[k].\position\) \textbf{then} \label{} \\
  \>\>\>\> \(\dt.\votes[k] \gets 1\) \` // record node's vote for the disputed transaction \label{} \\
  \>\>\>\> \(\dt.\yays \gets \dt.\yays + 1\) \label{} \\
  \>\>\> \textbf{else} \label{} \\
  \>\>\>\> \(\dt.\votes[k] \gets 0\) \` // record node's vote against the disputed transaction \label{} \\
  \>\>\>\> \(\dt.\nays \gets \dt.\nays + 1\) \label{cdend} \\
  \>\> \(\result.\disputes[\op{TxID}(\tx)] \gets \dt \) \label{storedt} \\
  \end{numbertabbing}
  }
  \caption{{\bf Ripple consensus protocol for node \iter (continued)}}
\end{algo*}

\begin{algo*}
  \vbox{
  \small
  \begin{numbertabbing}
  xxxx\=xxxx\=xxxx\=xxxx\=xxxx\=xxxx\=xxxx\kill
  // phase \establish \\
  \textbf{function} \op{updateOurProposals()}
  \` // Consensus.h:1361 \label{updateourproposals} \\
  \> \textbf{for} $j \in \unl$ \textbf{such that}  $(\op{clock.now()} - \currPeerPositions[j].\proposaltime) > 20 s$ \textbf{do} \` // Consensus.h:1378 \label{updateprop} \\
  \>\> $\currPeerPositions[\iterj] \gets \bot $ \` // remove stale proposals \label{} \\
  \> \(T \gets \result.\txns \) \` // current set of transactions, to update from disputed ones \label{emptysett} \\
  \> \textbf{for} $\dt \in \result.\disputes$ \textbf{do} \` // \dt is a disputed transaction \label{} \\
  \>\> \textbf{if} \op{updateVote(\dt)} \textbf{then}  \` // if vote on \dt changes, update the dispute set of the consensus round \label{} \\
  \>\>\> $\dt.\ourVote \gets \neg \dt.\ourVote$ \label{} \\
  \>\>\> \textbf{if} \dt.\ourVote \textbf{then} \` // should the transaction be included? DisputedTx.h:77 \label{} \\
  \>\>\>\> \(T \gets T \cup \{\dt.tx\}\) \` // \dt.\ourVote is initially set in \op{createDisputes(TxSet set)} \label{} \\
  \>\>\> \textbf{else} \label{} \\
  \>\>\>\> \(T \gets T \setminus \{\dt.tx\}\) \label{} \\
  \> \textbf{if} \(T \neq \result.\txns\)  \textbf{then} \` // if \txns changed, then update \result and tell the other nodes \label{} \\
  \>\> \(\result.\txns \gets T \) \label{} \\
  \> \> \(\result.\proposal \gets (\op{Hash(\prevledger)}, \result.\proposal.\seq + 1, \result.\txns, \iter)\) \label{} \\
  \>\> broadcast message \(\msg{proposal}{\result.\proposal}\) \label{broadcastnewproposal} \\
  \>\> \(\result.\disputes \gets \emptyhashmap\) \` // recompute \disputes after updating \result.\txns \label{} \\
  \>\> \textbf{for} \(j \in \unl \) \textbf{such that} $\currPeerPositions[j] \neq \bot$ \textbf{do}  \` // updateDisputes() at Consensus.h:1679 \label{updateDisputes} \\
  \>\>\> $\op{createDisputes}(\currPeerPositions[j].\position)$ \label{} \\
  \\
  \textbf{function} \op{updateVote(DisputedTx dt)} \` // DisputedTx.h:197 \label{updateVoteFunction} \\
  \> \textbf{if} \(\converge < 0.5\) \textbf{then} \` // set \threshold based on duration of the establish phase \label{thresholdstart} \\
  \>\>\(\threshold \gets 0.5\) \label{} \\
  \> \textbf{else if} \(\converge < 0.85\) \textbf{then} \label{} \\
  \>\>\(\threshold \gets 0.65\) \label{} \\
  \> \textbf{else if} \(\converge < 2\) \textbf{then} \label{} \\
  \>\> \(\threshold \gets 0.7\) \label{} \\
  \> \textbf{else} \label{} \\
  \>\>\(\threshold \gets 0.95\) \label{thresholdend} \\
  \> \(\newvote \gets \Bigl( \frac{\dt.\yays + \op{boolToInt}(\dt.\ourVote)}{\dt.\yays + \dt.\nays + 1} > \threshold \Bigr) \) \label{l:newvote}  \\
  \> \textbf{return} \(\bigl(\newvote \neq \dt.\ourVote\bigr)\) \` // the vote changes \label{l:comparevote} \\
  \\
  \textbf{function} \op{haveConsensus()} \` // Consensus.h:1545 \label{haveConsensusFunction} \\
  \> // count number of agreements and disagreements with our proposal \label{} \\
  \> \(\agree \gets |\{\iterj| \currPeerPositions[\iterj] =  \result.\proposal\}|\) \label{l:agree} \\
  \> \(\disagree \gets |\{\iterj| \currPeerPositions[\iterj] \neq \bot \land \currPeerPositions[\iterj] \neq  \result.\proposal\}|\) \label{l:disagree} \\
  \> \textbf{return} \( \bigl(\frac{\agree + 1 }{\agree + \disagree + 1} \ge 0.8 \bigr) \) \` // 0.8 is defined in ConsensusParams.h,  Consensus.cpp:104 \label{checkhaveconsensus}
  \end{numbertabbing}
  }
  \caption{{\bf Ripple consensus protocol for node \iter (continued)}}
\end{algo*}

\begin{algo*}
  \vbox{
  \small
  \begin{numbertabbing}
  xxxx\=xxxx\=xxxx\=xxxx\=xxxx\=xxxx\=xxxx\kill
  // phase \accepted\\
  \textbf{function} \op{onAccept()} \` // RCLConsensus.cpp:408 \label{onacceptfunction} \\
  \> \(\ledger \gets (\prevledger, \result.\txns) \) \` // \ledger is the last closed ledger, RCLConsensus.cpp:708 \label{l:buildlcl} \\
  \> $\validations[\iter] \gets \ledger$\label{} \\
  \> \(\val \gets \op{sign\textsubscript{\iter}(\ledger)}\) \` // validate the ledger, RCLConsensus.cpp:743 \label{l:sign} \\
  \> broadcast message \(\msg{validation}{\iter, \val, \ledger}\) \label{broadcastvalidation} \\
  \>\(\prevledger \gets \ledger\) \` // store the last closed ledger \label{l:prevledg} \\
  \>\(\prevRoundTime \gets \result.\roundTime\) \label{} \\
  \> \op{beginConsensus()} \` // advance to the next round of consensus, NetworkOPs.cpp:1584 \label{l:startconsensusagain} \\
  \\
  \textbf{upon} receiving a message \(\msg{validation}{j, \val, \ledger'}\) \textbf{such that} \` // LedgerMaster.cpp:858 \label{receiveval} \\
  \>\>  \(\ledger'.seq = \ledger.seq\) \textbf{and} \op{verify\textsubscript{\iterj}}\((\ledger',\val)\) \textbf{do} \label{} \\
  \> add $\ledger'$ to \tree \label{addtotree} \\
  \> $\validations[\iterj] \gets \ledger'$  \` // store received validation \label{} \\
  \> \(\valCount \gets |\{k \in \unl|\validations[k] = \ledger\}| \) \` // count the number of validations \label{l:valcount} \\
  \> \textbf{if} \(\valCount \ge 0.8 \cdot |\unl| \) \textbf{and} \(\ledger.seq > \fullyledger.seq\) \textbf{then} \label{l:checkquorum} \\
  \>\>\(\fullyledger \gets \ledger \) \` // ledger becomes fully validated \label{l:fullyledger} \\
  \>\>\(\s \gets \s \setminus \{\ledger.\txns\} \)  \label{} \\
  \>\> \textbf{for} \(\tx \in \ledger.\txns \) \textbf{do} \` // in some deterministic order \label{} \\
  \>\>\> \op{execute(\tx)} \label{executetxns} \\
  \\
  
  \textbf{function} \op{getPreferred(Ledger L)} \` // LedgerTrie.h:677 \label{getPreferredFunction} \\
  \> \textbf{if} \ledger is a leaf node in \tree \textbf{then} \label{} \\
  \>\> \textbf{return} \ledger \label{} \\
  \> \textbf{else} \label{} \\
  \>\> $ M \gets \argmax\{\support(N)\ |\ N\ \text{is a child of $L$ in the \tree}\} $ \label{} \\
  \>\> \textbf{if} $\uncommitted(M) \geq \support(M)$ \textbf{then}  \label{cond1} \\
  \>\>\> \textbf{return} \ledger \label{preferredL1} \\
  \>\> \textbf{else if} $ \max\{\support(N)\ |\ N\in siblings(M)\} + \uncommitted(M) < \support(M)$ \textbf{then} \label{cond2} \\
  \>\>\> \textbf{return} \op{getPreferred(M)}\label{getpreferredL2} \\
  \>\> \textbf{else} \label{} \\
  \>\>\> \textbf{return} \ledger \label{getPreferredEnd} 
  \end{numbertabbing}
  }
  \caption{{\bf Ripple consensus protocol for node \iter (continued)}}
  \label{alg:alg4}
\end{algo*}

\paragraph{Remarks on the pseudocode.} Next to every function name, a
comment points to a specific file and line in the source code which
contains its implementation. The Ripple source contains a large number of
files and most of the consensus protocol implementation is actually spread
over multiple header (.h) files, which complicates the analysis of the
code. The references in this work are based on version 1.4.0\footnote{The latest release (16.\ November 2020) is version 1.6.0. Compared to version 1.4.0, the current release has no significant changes concerning the consensus protocol.} of
\texttt{rippled}~\cite{sourcecode}.

\section{Violation of safety}
\label{sec:safety}

In this section we address the safety of the Ripple consensus protocol. We
first describe a simple scenario, in which consensus is violated in an
execution with seven nodes, of which one is Byzantine. Secondly, we show
how this problem can be generalized to executions with more nodes.

\subsection{Violating agreement with seven nodes}
\label{sec:safetyattack}

To show that the Ripple consensus protocol violates safety and may let two
correct nodes execute different transactions, we use the following scenario
with seven nodes. Figure~\ref{fig:safetyexample} gives a graphical
representation of our scenario and we will refer to it later in the text.

\begin{figure}[h!]
  \centering
  \includegraphics[width=0.6\textwidth]{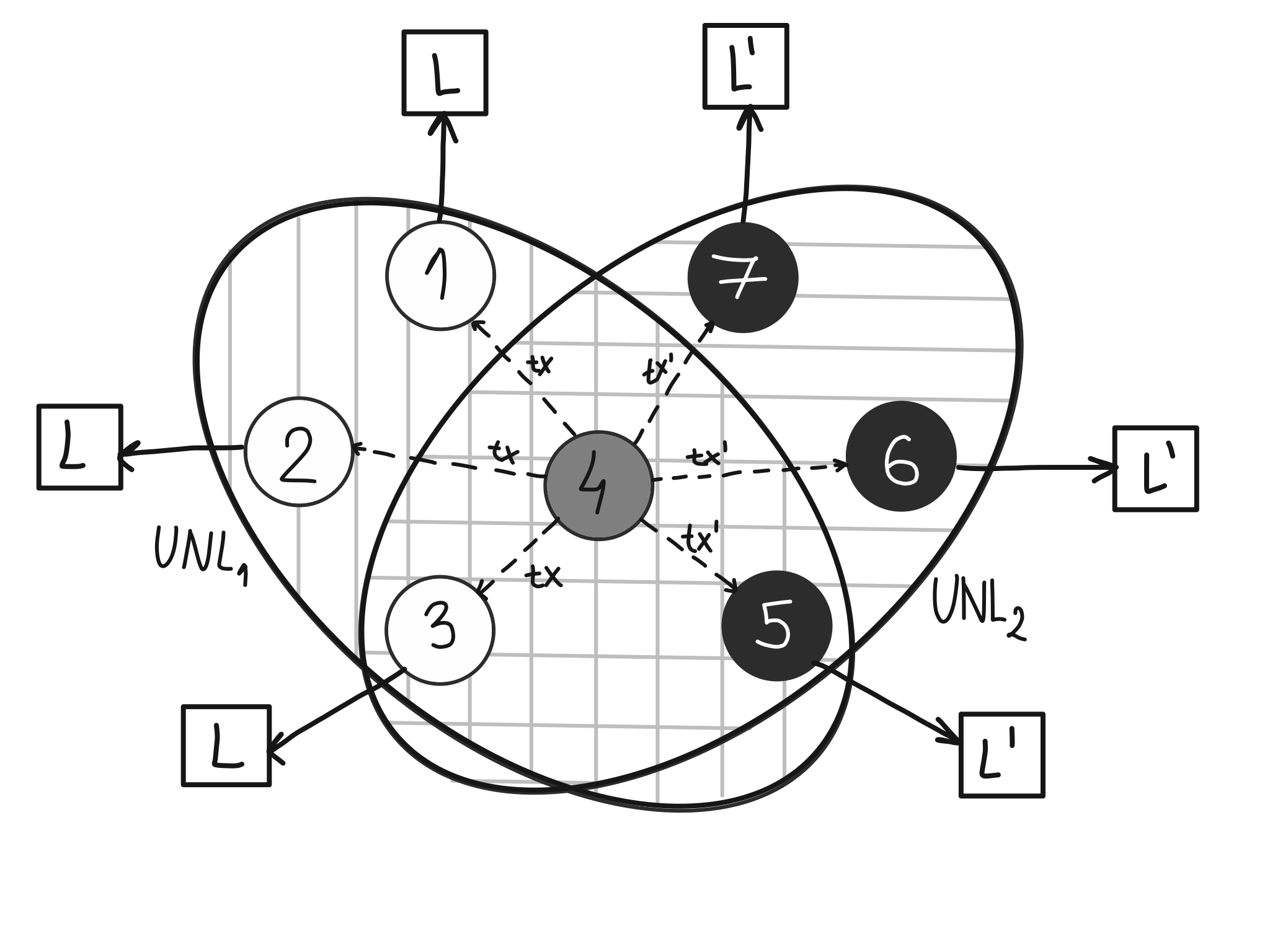}
  \caption{Example setup for showing a safety violation in the Ripple
    consensus protocol.  The setup consists of seven nodes, one of them
    Byzantine, and two UNLs.  Nodes 1, 2, and 3 (white) adopt $\unl_1$,
    vertically hatched, and nodes 5, 6, and 7 adopt $\unl_2$, horizontally
    hatched. Node~4 (gray) is Byzantine.}
  \label{fig:safetyexample}
\end{figure}

Nodes are named by numbers.  We let $\unl_1 = \{1, 2, 3, 4, 5\}$ and
$\unl_2 = \{3, 4, 5, 6, 7\}$, as illustrated by the two hatched areas in
the figure.  Nodes 1, 2, and 3 (white) trust $\unl_1$, nodes 5, 6, and 7
(black) trust $\unl_2$, and they are all correct; node~4 (gray) is
Byzantine.  With this setup, we achieve 60\% overlap between the UNLs of
any two nodes.

The key idea is that the Byzantine node (4) changes its behavior depending
on the group of nodes to which it communicates.  It will cause nodes 1, 2,
and 3 (white) to propose some transaction~\tx and nodes 5, 6, and 7 (black)
to propose a transaction~$\tx'$ for the next ledger. No other transaction
exists.  The Byzantine node (4) follows the protocol as if it had proposed
\tx when interacting with the white nodes and behaves as if it had proposed
$\tx'$ when interacting with the black nodes.

Assuming that all nodes start the consensus roughly at the same time and they do not switch the preferred ledger, the
protocol does the following:
\begin{itemize}
\item The Byzantine node 4 submits $\tx$ and $\tx'$ using gossip and causes
  \msg{submit}{\tx} to be received by nodes 1, 2, and 3 and
  \msg{submit}{\tx'} to be received by nodes 5, 6, and 7 from the gossip
  layer.  During the repeated \heartbeat timer executions in the \open
  phase, all correct nodes have the same value of \prevledger and send no
  further messages.
\item Suppose at a common execution of the \heartbeat timer execution
  (L\ref{heartbeat}) all correct nodes proceed to the \establish phase and
  call \op{closeLedger}.  They broadcast the message \msg{proposal}{\s},
  with $\s$ containing $\tx$ or $\tx'$, respectively
  (L\ref{broadcastproposal1}).  Node~4 sends a \str{proposal} message
  containing $\tx$ to nodes~1, 2, and 3 and one containing $\tx'$ to
  nodes~5-7.  Furthermore, every correct node executes
  $\op{createDisputes}$ with the transaction set~\txns received in
  each \str{proposal} message, which creates \result.\disputes
  (L\ref{createDisputesFunction}).  For nodes 1, 2, and 3, transaction
  $\tx'$ is disputed and for nodes 5, 6, and 7, transaction \tx is
  disputed.
\item During \establish phase, all nodes update their vote for each
  disputed transaction (L\ref{updateVoteFunction}). Nodes~1, 2, and 3
  consider $\tx'$ but do not change their \emph{no} vote on~$\tx'$ because
  only 20\% of nodes in their UNL (namely, node~5) vote \emph{yes} on
  $\tx'$; this is less than required \threshold of 50\% or more
  (L\ref{l:newvote}).  The same holds for nodes~5, 6, and 7 with respect to
  transaction $\tx$.  Hence, \result.\txns remains unchanged and no correct
  node sends another \str{proposal} message.
\item Eventually, function \op{haveConsensus} returns \true for each
  correct node because the required $4/5 = 80\%$ of its \unl has issued the
  same proposal as the node itself (L\ref{haveConsensusFunction}).  Every
  correct node moves to the \accepted phase.
\item During \op{onAccept}, nodes~1, 2, and 3 send a \str{validation}
  message with ledger $L = (\prevledger, \{\tx\})$, whereas nodes~5, 6, and
  7 send a \str{validation} message containing
  $L' = (\prevledger, \{\tx'\})$ (L\ref{broadcastvalidation}).  Node~4
  sends a \str{validation} message containing $\tx$ to nodes~1, 2, and 3
  and a different one, containing $\tx'$, to nodes~5, 6, and 7.
\item Every correct node subsequently receives five validation messages,
  from all nodes in its UNL, and finds that 80\% among them contain the
  same ledger (L\ref{receiveval}). Observe that no node changes its preferred ledger after calling \op{getPreferred}. This implies that nodes 1, 2, and 3
  fully validate $L$ and execute \tx, whereas nodes 5, 6, and 7 fully
  validate $L'$ and execute $\tx'$.  Hence, the agreement condition of
  consensus is violated.
\end{itemize}

\subsection{Generalization}

We now generalize the previous scenario and show a violation of agreement
with an arbitrarily large number of nodes.  As illustrated in
Figure~\ref{fig:gensafety}, the system consists of $M = 2n+f$ nodes, such
that nodes~$1, \dots, n$ (white) each submit transaction~$\tx$,
nodes~$n+1, \dots, n+f$ (gray) are Byzantine, and
nodes~$n+f+1, \dots, 2n + f$ (black) each submit transaction~$\tx'$.
Assume all correct nodes have one of two different UNLs, namely
$\unl_{\tx} = \{1, \dots, n + f + \tilde{n}\}$ or
$\unl_{\tx'} = \{n - \tilde{n} + 1, \dots, 2n + f\}$, each of size
$n + f + \tilde{n}$.  As the names suggest, nodes $1, \dots, n$, which
submit~$\tx$, use $\unl_{\tx}$ and nodes $n+f+1, \dots, 2n+f$, which
submit~$\tx'$, use $\unl_{\tx'}$.

The execution proceeds analogously to the one in the previous section, with
nodes $1, \dots, n$ behaving like nodes~1, 2, and 3, the $f$ Byzantine
nodes here behaving like node~4, and nodes $n+f+1, \dots, 2n+f$ behaving
like nodes~5, 6, and 7.  The strategy of the Byzantines nodes is to follow
the protocol, as if they had submitted transaction $\tx$ when they interact
with correct nodes $1, \dots, n$, and to behave as if they had submitted
transaction $\tx'$ when they interact with correct
nodes~$n+f+1, \dots, 2n+f$.

\begin{theorem}\label{thm:safety}
  A system of $2n+f$ nodes, of which $f$ are Byzantine, running the Ripple
  consensus protocol according to the scenario defined above may violate
  safety if
  \begin{equation}
    \frac{n+f}{n+\tilde{n}+f} \geq 0.8 \label{eq:general1}.
  \end{equation}
\end{theorem}

\begin{proof}
  To prove that safety can be violated, it is enough to show that the
  strategy of the Byzantine nodes is successful.  This follows from the
  same argument as in the previous scenario with seven nodes, according to
  the pseudocode in Section~\ref{sec:analysis}.  The
  condition~(\ref{eq:general1}) is corresponds to the test for fully
  validating a ledger (L\ref{l:checkquorum}).
\end{proof}

\begin{figure}
  \centering
  \includegraphics[width=0.7\textwidth]{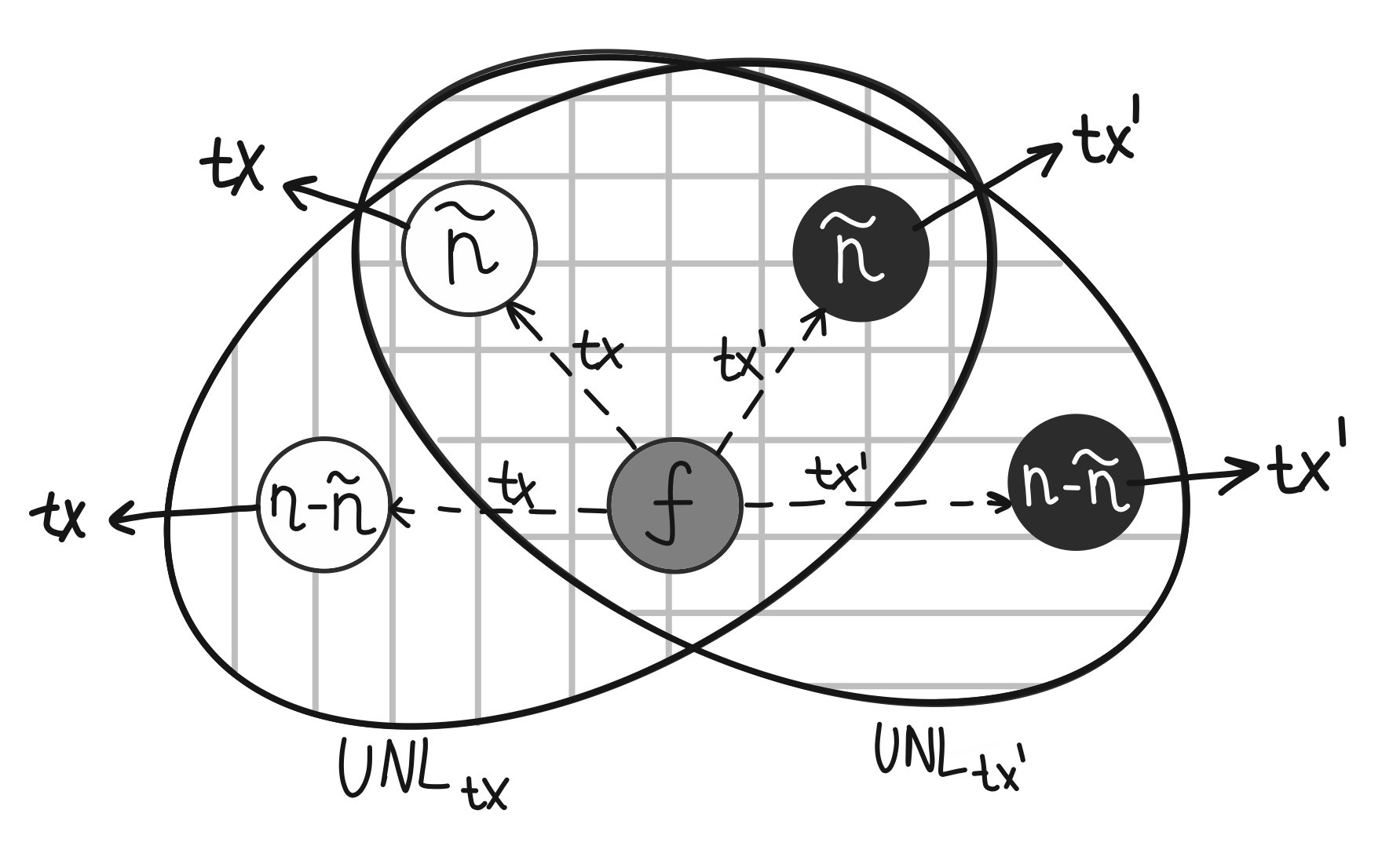}
  \caption{The generalized attack scenario with $2n+f$ nodes.  The $n$
    white nodes submit~$\tx$ and have $\unl_{\tx}$, while the $n$ black
    nodes submit $\tx'$ instead and have $\unl_{\tx'}$.  The $f$ Byzantine
    nodes (gray) behave differently, depending on whether they interact
    with white and black nodes, respectively.}
  \label{fig:gensafety}
\end{figure}

The bound (\ref{eq:general1}) of the theorem corresponds directly to the
condition in the source code.  We can reformulate this, using
$\ol=\frac{2\tilde{n}+f}{n+\tilde{n}+f}$ to denote the relative overlap of
the UNLs (i.e., the fraction of nodes that are common between the two
UNLs).

\begin{corollary}\label{cor:safety}
  The Ripple consensus protocol may violate safety in a system of $2n+f$
  nodes if
  \begin{equation}\label{eq:cor:safety}
    f \ \geq \ n \: \frac{5\ol-2}{6-5\ol}
  \end{equation}
  or equivalently, recalling that the total number of correct nodes is
  $2n$,
  \begin{equation}\label{eq:cor:safety2}
    f \ \geq \ 2n \: \frac{5 \ol-2}{12-10\ol}.
  \end{equation}
\end{corollary}

\begin{proof}
  Equation~(\ref{eq:cor:safety}) follows directly from (\ref{eq:general1})
  by substituting $\tilde{n}$ in terms of the overlap~\ol.  Furthermore,
  (\ref{eq:general1}) follows from (\ref{eq:cor:safety}) by replacing the
  \unl size through the total number of nodes.
\end{proof}

Corollary~\ref{cor:safety} illustrates the number of Byzantine nodes
required to break the safety of the protocol using the presented strategy.
The number of Byzantine nodes required to show the violation is
proportional to $n$, the number of correct nodes.

\section{Violation of liveness}
\label{sec:liveness}

In this section, we show how the liveness of the Ripple consensus protocol
may be violated, even when all nodes have the same UNL and only one node is
Byzantine.  One can bring the protocol to a state, in which it cannot
produce a correct ledger and where it stops making progress.

Consider a system with $2n$ correct nodes and one single Byzantine node.
All nodes are assumed to trust each other, i.e., there is one common UNL
containing all $2n+1$ nodes.  Observe that in this system, the fraction of
Byzantines nodes can be made arbitrary small by increasing~$n$.

\begin{figure}
  \centering
  \includegraphics[width=0.7\textwidth]{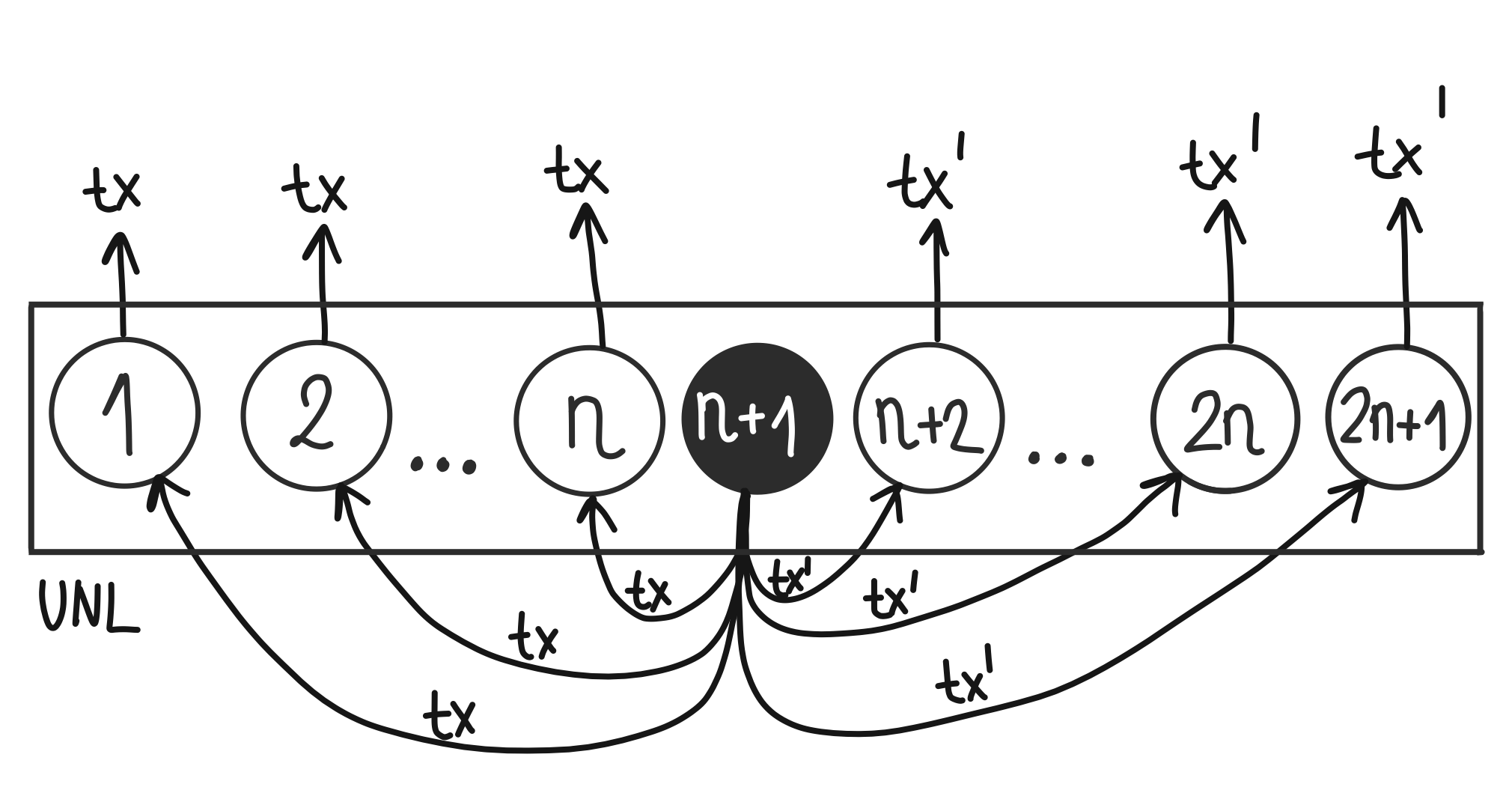}
  \caption{Setup in which liveness is violated in the Ripple network. The
    network consists of $2n+1$ nodes with one single UNL and 1 Byzantine
    (black). The $n$ first nodes propose transaction $\tx$ while the last
    $n$ propose transaction $\tx'$. The Byzantine proposes transaction
    $\tx$ to the $n$ first nodes and transaction $\tx'$ to the last $n$.}
  \label{fig:livenessexample}
\end{figure}

As illustrated in Figure~\ref{fig:livenessexample}, node~$n+1$, which is
Byzantine, exhibits a split-brain behavior and follows the protocol for an
input transaction~\tx when interacting with nodes $1, \dots, n$, and
operates with a different input transaction~$\tx'$ when interacting with
nodes $n+2, \dots, 2n+1$.  This implies that the first half of the correct
nodes, denoted $1,\dots,n$, will propose a transaction~$\tx$ and the other
half, nodes $n+2, \dots, 2n+1$, will propose transaction $\tx'$.  Similar
to the execution shown in Section~\ref{sec:safetyattack}, the nodes start
the consensus protocol roughly at the same time and they do not switch the preferred ledger, they proceed like this:

\begin{itemize}
\item Byzantine node~$n+1$ sends two messages, \msg{submit}{\tx} and
  \msg{submit}{\tx'}, using the gossip layer and causes \tx to be received
  by the first $n$ correct nodes and $\tx'$ to be received by the last $n$
  correct nodes.
\item After some time has passed, the correct nodes start to close the
  ledger and move to the \establish phase.  Every correct node sends a
  \str{proposal} message, containing only the submitted transaction of
  which it knows (L\ref{broadcastproposal1}), namely \tx for the first $n$
  correct nodes and $\tx'$ for the last $n$ correct nodes.
\item During \establish phase, the correct nodes receive the \str{proposal}
  messages from all nodes (including the Byzantine node) and store them in
  \currPeerPositions (L\ref{receiveprop1}).  Since they all use the same
  UNL, all obtain the same \str{proposal} messages from the correct nodes.
\item  Each node creates disputes (L\ref{createDisputesFunction})
  and updates them while more \str{proposal} messages arrive.
  Since the proposed transaction sets differ,  each node creates
  a dispute for $\tx$ and for~$\tx'$.
\item While the \str{proposal} messages are being processed, votes are
  counted in \op{updateVotes} (L\ref{updateVoteFunction}), using the \yays
  and \nays of each disputed transaction.  For a correct node in
  $\{1, \dots, n\}$, notice that the first $n$ nodes and the Byzantine node
  vote \emph{no} for $\tx'$ and the last $n$ nodes vote \emph{yes}.  Thus,
  the fraction of nodes voting \emph{yes} for $\tx'$ is less than required
  threshold (50\%), and so the first $n$ nodes continue to vote \emph{no}
  for $\tx'$.  Similarly, nodes $n+2$ to $2n+1$ never update their vote on
  \tx and always vote \emph{no} for \tx.
\item The \op{haveConsensus} function called periodically during the
  \establish phase checks if at least 80\% of the nodes in the UNL agree on
  the the proposal of the node itself (L\ref{haveConsensusFunction}).  From
  the perspective of each one of the first $n$ correct nodes, $n$ other
  nodes agree and $n$ nodes disagree with its proposal, which contains \tx.
  That is not enough support for achieving consensus and the function
  will return \false.  The same holds from the perspective of the last $n$
  correct nodes, which also continuously return \false.
\item Finally, the correct nodes will continue trying to update votes and
  get enough support, but without being able to generate a correct ledger.
  No correct node proceeds to validating the ledger.  In other words,
  liveness of the protocol is not guaranteed.
\end{itemize}


\section{Conclusion}
\label{sec:conclusion}

Ripple is one of the oldest public blockchain platforms. For a long time, its native XRP token has been the third-most valuable in terms of its total market capitalization. The Ripple network is implemented as a peer-to-peer network of validator nodes, which should reach consensus even in the presence of faulty or malicious nodes. Its consensus protocol is generally considered to be a Byzantine fault-tolerant protocol, but without global knowledge of all participating nodes and where a node only communicates with other nodes it knows from its UNL.

Previous work regarding the Ripple consensus protocol has already brought up some concerns about its liveness and safety. In order to better analyze the protocol, this work has presented an independent, abstract description derived directly from the implementation. Furthermore, this work has identified relatively simple cases, in which the protocol may violate safety and/or liveness and which have devastating effects on the health of the network. Our analysis illustrates the need for very close synchronization, tight interconnection, and fault-free operations among the participating validators in the Ripple network.

\section*{Acknowledgments}
The authors thank Luca Zanolini and Orestis Alpos for interesting discussions.

This work has been funded by the Swiss National Science Foundation (SNSF)
under grant agreement Nr\@.~200021\_188443 (Advanced Consensus Protocols).

\bibliography{references}

\begin{thebibliography}{10}

\bibitem{alpern1985defining}
B.~Alpern and F.~B. Schneider, ``Defining liveness,'' {\em Inf. Process.
  Lett.}, vol.~21, no.~4, pp.~181--185, 1985.

\bibitem{abbccd18}
E.~Androulaki, A.~Barger, V.~Bortnikov, C.~Cachin, K.~Christidis, A.~D. Caro,
  D.~Enyeart, C.~Ferris, G.~Laventman, Y.~Manevich, S.~Muralidharan, C.~Murthy,
  B.~Nguyen, M.~Sethi, G.~Singh, K.~Smith, A.~Sorniotti, C.~Stathakopoulou,
  M.~Vukolic, S.~W. Cocco, and J.~Yellick, ``Hyperledger fabric: a distributed
  operating system for permissioned blockchains,'' in {\em Proceedings of the
  Thirteenth EuroSys Conference, EuroSys 2018, Porto, Portugal, April 23-26,
  2018} (R.~Oliveira, P.~Felber, and Y.~C. Hu, eds.), pp.~30:1--30:15, {ACM},
  2018.

\bibitem{armknecht2015ripple}
F.~Armknecht, G.~O. Karame, A.~Mandal, F.~Youssef, and E.~Zenner, ``Ripple:
  Overview and outlook,'' in {\em Trust and Trustworthy Computing - 8th
  International Conference, {TRUST} 2015, Heraklion, Greece, August 24-26,
  2015, Proceedings} (M.~Conti, M.~Schunter, and I.~G. Askoxylakis, eds.),
  vol.~9229 of {\em Lecture Notes in Computer Science}, pp.~163--180, Springer,
  2015.

\bibitem{bukwmi18}
E.~Buchman, J.~Kwon, and Z.~Milosevic, ``The latest gossip on {BFT}
  consensus,'' {\em CoRR}, vol.~abs/1807.04938, 2018.

\bibitem{CachinGR11}
C.~Cachin, R.~Guerraoui, and L.~E.~T. Rodrigues, {\em Introduction to Reliable
  and Secure Distributed Programming {(2.} ed.)}.
\newblock Springer, 2011.

\bibitem{cactac19}
C.~Cachin and B.~Tackmann, ``Asymmetric distributed trust,'' in {\em 23rd
  International Conference on Principles of Distributed Systems, {OPODIS} 2019,
  December 17-19, 2019, Neuch{\^{a}}tel, Switzerland} (P.~Felber, R.~Friedman,
  S.~Gilbert, and A.~Miller, eds.), vol.~153 of {\em LIPIcs}, pp.~7:1--7:16,
  Schloss Dagstuhl - Leibniz-Zentrum f{\"{u}}r Informatik, 2019.

\bibitem{cacvuk17}
C.~Cachin and M.~Vukolic, ``Blockchain consensus protocols in the wild (keynote
  talk),'' in {\em 31st International Symposium on Distributed Computing,
  {DISC} 2017, October 16-20, 2017, Vienna, Austria} (A.~W. Richa, ed.),
  vol.~91 of {\em LIPIcs}, pp.~1:1--1:16, Schloss Dagstuhl - Leibniz-Zentrum
  f{\"{u}}r Informatik, 2017.

\bibitem{caslis02}
M.~Castro and B.~Liskov, ``Practical byzantine fault tolerance and proactive
  recovery,'' {\em {ACM} Trans. Comput. Syst.}, vol.~20, no.~4, pp.~398--461,
  2002.

\bibitem{CharronBostPS10}
B.~Charron{-}Bost, F.~Pedone, and A.~Schiper, eds., {\em Replication: Theory
  and Practice}, vol.~5959 of {\em Lecture Notes in Computer Science}.
\newblock Springer, 2010.

\bibitem{chase2018analysis}
B.~Chase and E.~MacBrough, ``Analysis of the {XRP} ledger consensus protocol,''
  {\em CoRR}, vol.~abs/1802.07242, 2018.

\bibitem{dwlyst88}
C.~Dwork, N.~A. Lynch, and L.~J. Stockmeyer, ``Consensus in the presence of
  partial synchrony,'' {\em J. {ACM}}, vol.~35, no.~2, pp.~288--323, 1988.

\bibitem{sbft}
G.~Golan{-}Gueta, I.~Abraham, S.~Grossman, D.~Malkhi, B.~Pinkas, M.~K. Reiter,
  D.~Seredinschi, O.~Tamir, and A.~Tomescu, ``{SBFT:} {A} scalable and
  decentralized trust infrastructure,'' in {\em 49th Annual {IEEE/IFIP}
  International Conference on Dependable Systems and Networks, {DSN} 2019,
  Portland, OR, USA, June 24-27, 2019}, pp.~568--580, {IEEE}, 2019.

\bibitem{hadtou93}
V.~Hadzilacos and S.~Toueg, ``Fault-tolerant broadcasts and related problems,''
  in {\em Distributed Systems (2nd Ed.)} (S.~J. Mullender, ed.), New York: ACM
  Press \& Addison-Wesley, 1993.

\bibitem{lashpe82}
L.~Lamport, R.~E. Shostak, and M.~C. Pease, ``The byzantine generals problem,''
  {\em {ACM} Trans. Program. Lang. Syst.}, vol.~4, no.~3, pp.~382--401, 1982.

\bibitem{libra20librabft}
{LibraBFT Team}, ``State machine replication in the {Libra} blockchain.''
  Technical report, 2020.
\newblock
  \url{https://developers.libra.org/docs/state-machine-replication-paper}.

\bibitem{logama19}
G.~Losa, E.~Gafni, and D.~Mazi{\`{e}}res, ``Stellar consensus by
  instantiation,'' in {\em 33rd International Symposium on Distributed
  Computing, {DISC} 2019, October 14-18, 2019, Budapest, Hungary} (J.~Suomela,
  ed.), vol.~146 of {\em LIPIcs}, pp.~27:1--27:15, Schloss Dagstuhl -
  Leibniz-Zentrum f{\"{u}}r Informatik, 2019.

\bibitem{lumest17}
A.~D. Luzio, A.~Mei, and J.~Stefa, ``Consensus robustness and transaction
  de-anonymization in the ripple currency exchange system,'' in {\em 37th
  {IEEE} International Conference on Distributed Computing Systems, {ICDCS}
  2017, Atlanta, GA, USA, June 5-8, 2017} (K.~Lee and L.~Liu, eds.),
  pp.~140--150, {IEEE} Computer Society, 2017.

\bibitem{mauri2020formal}
L.~Mauri, S.~Cimato, and E.~Damiani, ``A formal approach for the analysis of
  the {XRP} ledger consensus protocol,'' in {\em Proceedings of the 6th
  International Conference on Information Systems Security and Privacy,
  {ICISSP} 2020, Valletta, Malta, February 25-27, 2020} (S.~Furnell, P.~Mori,
  E.~R. Weippl, and O.~Camp, eds.), pp.~52--63, {SCITEPRESS}, 2020.

\bibitem{mmskf18}
P.~Moreno{-}Sanchez, N.~Modi, R.~Songhela, A.~Kate, and S.~Fahmy, ``Mind your
  credit: Assessing the health of the ripple credit network,'' in {\em
  Proceedings of the 2018 World Wide Web Conference on World Wide Web, {WWW}
  2018, Lyon, France, April 23-27, 2018} (P.~Champin, F.~L. Gandon, M.~Lalmas,
  and P.~G. Ipeirotis, eds.), pp.~329--338, {ACM}, 2018.

\bibitem{bitcoin}
S.~Nakamoto, ``Bitcoin: A peer-to-peer electronic cash system.'' Whitepaper,
  2009.
\newblock \url{http://bitcoin.org/bitcoin.pdf}.

\bibitem{peshla80}
M.~C. Pease, R.~E. Shostak, and L.~Lamport, ``Reaching agreement in the
  presence of faults,'' {\em J. {ACM}}, vol.~27, no.~2, pp.~228--234, 1980.

\bibitem{sourcecode}
{Ripple Labs}, ``Ripple 1.4.0.''
  \url{https://github.com/ripple/rippled/releases/tag/1.4.0}.

\bibitem{rippleledgeroverview20}
{Ripple Labs}, ``{XRP Ledger Documentation $>$ Concepts $>$ Introduction $>$
  XRP Ledger Overview}.'' Available online,
  \url{https://xrpl.org/xrp-ledger-overview.html}.

\bibitem{rippledoc20}
{Ripple Labs}, ``{XRP Ledger Documentation $>$ Concepts $>$ Consensus Network
  $>$ Consensus}.'' Available online, \url{https://xrpl.org/consensus.html},
  2020.

\bibitem{schnei90}
F.~B. Schneider, ``Implementing fault-tolerant services using the state machine
  approach: {A} tutorial,'' {\em {ACM} Comput. Surv.}, vol.~22, no.~4,
  pp.~299--319, 1990.

\bibitem{schwartz2014ripple}
D.~Schwartz, N.~Youngs, and A.~Britto, ``The {Ripple} protocol consensus
  algorithm.'' Ripple Labs Inc., available online,
  \url{https://ripple.com/files/ripple_consensus_whitepaper.pdf}, 2014.

\end{thebibliography}
\bibliographystyle{ieeesort}

\end{document}